\documentclass[12pt, a4paper]{llncs}
\usepackage{inputenc} 
\usepackage{fontenc}
\usepackage{url}              
\usepackage{cite}             

\usepackage{amsmath,bm}  
\usepackage{graphicx}         

\usepackage{multirow}
\usepackage{xcolor}
\usepackage{xstring}
\usepackage{caption}
\usepackage[english]{babel}
\usepackage{amsfonts}
\usepackage{amssymb}
\usepackage{theoremref}
\usepackage{ifthen}
\usepackage[linesnumbered,algoruled,boxed,lined]{algorithm2e}

\labelformat{algocf}{\textbf{Alg.\,#1}}

\def\BibTeX{{\rm B\kern-.05em{\sc i\kern-.025em b}\kern-.08em
    T\kern-.1667em\lower.7ex\hbox{E}\kern-.125emX}}

\newcommand{\F}{\mathbb{F}}
\newcommand{\Fq}{\mathbb{F}_{q}}
\newcommand{\Fqm}{\mathbb{F}_{q^m}}

\newcommand{\bb}{\mathbf{b}}
\newcommand{\bc}{\mathbf{c}}
\newcommand{\bs}{\mathbf{s}}
\newcommand{\be}{\mathbf{e}}

\newcommand{\bx}{\mathbf{x}}
\newcommand{\by}{\mathbf{y}}

\newcommand{\bu}{\mathbf{u}}
\newcommand{\bv}{\mathbf{v}}

\newcommand{\bzero}{\mathbf{0}}
\newcommand{\balpha}{\bm{\alpha}}
\newcommand{\bbeta}{\bm{\beta}}

\newcommand{\dr}{\mathrm{d_R}}
\newcommand{\calA}{\mathcal{A}}
\newcommand{\calB}{\mathcal{E}}
\newcommand{\calC}{\mathcal{C}}

\newcommand{\calE}{\mathcal{E}}

\newcommand{\calR}{\mathcal{R}}
\newcommand{\calS}{\mathcal{S}}
\newcommand{\calT}{\mathcal{T}}

\DeclareMathOperator{\Mul}{Mul}

\counterwithout{table}{section}

\newcommand{\Enum}[3]{#1_{#2}, \ldots, #1_{#3}}

\newcommand{\Span}[2]{\langle #2 \rangle_{#1}}

\title{Two new algorithms for error support recovery of  low rank parity check codes
} 

\begin{document}
 \author{Ermes Franch\inst{1}, Chunlei Li\inst{1}}
 \institute{\inst{1}University of Bergen, Norway \\
 \email{\{ermes.franch, chunlei.li\}@uib.no} \thanks{This work was supported by the Research Council of Norway under Grant No.~311646/O70}}

\maketitle

\begin{abstract} 
   Due to their weak algebraic structure, low rank parity check (LRPC) codes have been employed in several post-quantum cryptographic schemes.
  In this paper we propose new improved decoding algorithms for $[n,k]_{q^m}$ LRPC codes of dual rank weight $d$.
  The proposed algorithms
  can
  efficiently decode LRPC codes
  with the parameters satisfying $n-k=rd-c$, where
  $r$ is the dimension of the error support and $c\leq d-2$.
  They outperform the original decoding algorithm of LRPC codes when $d>2$ and allow for decoding LRPC codes with a higher code rate and smaller values $m$.
\end{abstract}

\section{Introduction}

Rank-metric codes, which are embedded in a rank metric space,
have applications in network coding \cite{KoetterKschischang}, space-time codes \cite{Gabidulin-SpaceTimeCodes}, distributed storage \cite{RankBook1}, and cryptography \cite{RankBook2, GPT, gaborit2013, RankSign2014, gaborit2017, GaboritSurvey, DuranteSiciliano, ROLLO}. 
Rank-based cryptography relies on the difficulty of the rank syndrome decoding (RSD) problem.
So far the best-known method of solving the RSD problem has an exponential complexity which is quadratic in the parameter size \cite{GaboritZemor2016, GaboritRuattaSchrek2016}.
This nice feature allows for smaller sizes of keys in rank-based cryptosystems to achieve the same level of security provided by those cryptosystems based on Hamming-metric codes. 
Existing rank-based cryptographic schemes mainly used two types of rank-metric codes: Gabidulin codes \cite{Gabidulin1985} and low rank parity-check (LRPC) codes \cite{gaborit2013} and their variants, for which efficient decodings have been extensively studied \cite{Loidreau:2006aa,wachter2013fast,Kadir-li20}.
Due to the strong algebraic structure of Gabidulin codes, the GPT cryptosystem and its variants are subject to the structural attack by Overbeck \cite{Overbeck2008}.
LRPC codes can be seen as the equivalent of LDPC codes in rank metric and have a very weak algebraic structure.
These codes could be masked more easily in cryptosystems. Consequently, different schemes based on LRPC codes were proposed in recent years, RankSign \cite{RankSign2014}, ROLLO \cite{ROLLO}, Durandal \cite{DuranteSiciliano}, etc. 
On the other hand, without a certain algebraic structure, LRPC codes can only be decoded in a probabilistic manner. 
The original decoding algorithm of LPRC codes in \cite{gaborit2013} works only for the cases where the support of the syndrome is exactly the product space of the parity-check support and the error support.
When decoding LRPC codes, the error support recovery step contributes a dominating factor to decoding failures. 
In the extended paper \cite{LRPC-2019-TIT} on LRPC codes and their cryptographic applications, the authors further considered the cases where the syndrome support has a dimension $rd-c$ with $c<r$, where $r$ is the dimension of the error support, $d$ is the dual rank weight of the LRPC codes and $rd$ is the dimension of the product space between the two supports.
By applying two expansion functions on the syndrome support, they proposed new decoders that can correct errors with higher rank weights and decrease the decoding failure rate. 

In this paper we consider an alternative approach to decoding LRPC codes for the cases where the syndrome support has a dimension $rd-c$ with $c<d$. The proposed decoders rely on a crucial observation that employs all the elements (instead of only the basis elements) in the parity-check support, which enables us to significantly loosen the restriction on $m$ as required in \cite{LRPC-2019-TIT}.  The paper is organized as follows: Section \ref{Sec:pre} introduces notations, basics on rank metric codes, the problems of rank syndrome decoding and error support recovery.  Section \ref{Sec:LRPC} recalls the LRPC codes and their decoding approach. In Section \ref{Sec:decoder} we start with some theoretical analysis and then proposed two new algorithms for LRPC codes where
the syndrome support has dimension $rd-c$ with $c\leq d-2$; and Section \ref{Sec:dis} 
discusses the decoding failure rate of the proposed algorithms and their connections to the improved decoding algorithms in \cite{LRPC-2019-TIT}.

\section{Preliminaries} \label{Sec:pre}
We denote by $\Fqm$ the finite field of $q^m$ elements where $q$ is a prime power.
Vectors will be indicated by bold lower-case letters, and the $i$-th component of a vector will be indicated by the same letter in normal font, for example, $\bv = (\Enum{v}{1}{n})$. The notation $\Fq^n$ denotes the vector space of all the vectors of length $n$ over $\Fq$. 
A matrix $M$ will be indicated by upper-case letter, and the $(i,j)$-th entry of $M$ will be indicated by $m_{i,j}$.
The notation $\Fq^{m \times n}$ indicates all the $m\times n$ matrices over $\Fq$. The notation $[n]$ indicates the interval $\{1,\ldots, n\} \subset \mathbb{N}$.
Given a set $S \subseteq \Fqm$ and an element $a \in \Fqm$ the notation $Sa$ corresponds to the set $\{sa \mid s \in S\}$.

The field $\Fqm$ can be regarded as a vector space of dimension $m$ over the field $\Fq$.
We will denote $\Fq$-linear subspaces $\calS$ of $\Fqm$ by upper case calligraphic letters.
Given a set $S \subseteq \Fqm$, we denote by $\Span{\Fq}{S}$ the $\Fq$-linear subspace generated by the elements of this space.
For a given vector $\bv \in \Fqm^n,$ we denote as $\Span{\Fq}{\bv} = \Span{\Fq}{\Enum{v}{1}{n}}$ the subspace generated by the components of $\bv$, this will be called the support of $\bv$.
Similarly, for a matrix $M \in \Fqm^{k \times n}$, we will denote the subspace generated by all its entries as $\Span{\Fq}{M} = \Span{\Fq}{m_{i,j} \mid (i,j) \in [k]\times[n]}$.
The notations $\calA^n$ and $\calA^{m \times n}$ stand for the set of all the vectors of length $n$ having support $\calA$ and the set of all the $m \times n$ matrices with support $\calA$, respectively.

Using the notion of support, we can define the rank distance over $\Fqm^n$.
Consider $\bu,\bv \in \Fqm^n$, the \textbf{rank-weight} of $\bu$ is given by $\mathrm{w_R}(\bu) = \dim(\Span{\Fq}{\bu}),$ and 
the \textbf{rank distance} between two vectors is defined as
$
    \dr(\bu,\bv) = \mathrm{w_R}(\bu -\bv).
$
It can be proved that the function $\dr$ is a distance in the mathematical sense.


\begin{definition}
    A \textbf{rank metric code} is a subset $\calC \subseteq \Fqm^{n}$ equipped with the rank distance.
    The \textbf{minimum rank distance} of the code $\calC$ is given by the minimum distance between any two different elements of the code, i.e.,
    $
        \dr(\calC) = \min(\{\dr(\bu,\bv) \mid \bu\neq \bv \in \calC \}).
    $ If $\calC \subseteq \Fqm^n$ is an $\Fqm$-linear subspace of $\Fqm$ we call $\calC$ an \textbf{$\Fqm$-linear rank metric code}.
\end{definition}

\smallskip

An $\Fqm$-linear rank metric code can be defined by the use of a generator matrix or the use of a parity-check matrix.

\begin{definition}
    Let $\calC \subseteq \Fqm$ be a vector subspace of $\Fqm^n$ of dimension $k$.
    A \textbf{parity check matrix} of $\calC$ is a matrix $H \in \Fqm^{(n-k) \times n}$ of rank $n-k$ such that all the elements $\bc \in \calC$ satisfy $\bc H^\intercal = \bzero$. 
    For a generic $\bv \in \Fqm^n,$ we will have that $\bv H^\intercal = \bs \in \Fqm^{n-k}$, which is known as the \textbf{syndrome} of $\bv$. 
\end{definition}

One interesting problem in coding theory is the syndrome decoding problem. Below we recall this problem 
in the context of rank metric codes and a closely-related problem.

\begin{definition}
Given $H \in \Fqm^{n-k \times n}$ a parity check matrix of a $\Fqm$-linear rank metric code $\calC$, a syndrome $\by H^{\intercal} =  \bs \in \Fqm^{n-k}$ and a small integer $r$. 
The \textbf{Rank Syndrome Decoding (RSD) problem} consists in finding $\be \in \Fqm^n$ such that
    $$
        \be H^\intercal = \bs, \quad \mathrm{w_R}(\be) \le r.
    $$
\end{definition}

\begin{definition}
    Given $H \in \Fqm^{n-k \times n}$ a parity check matrix of an $\Fqm$-linear rank metric code $\calC$, a syndrome $\by H^{\intercal} = \bs \in \Fqm^{n-k}$ and a small integer $r$. 
    The \textbf{Rank Support Recovery problem} is to find a subspace $\calB \subseteq \Fqm$ of dimension $\le r$ such that there exists $\be \in \calB^n$ and
    $
        \be H^\intercal = \bs.
    $
\end{definition}
A vector $\be \in \Fqm$ of rank $r$ will have a support $\calB = \Span{\Fq}{\be} = \Span{\Fq}{\Enum{\beta}{1}{r}}$ where $\bbeta = (\Enum{\beta}{1}{r})$ is a basis of $\calB$.
The vector $\be$ can then be represented as $\be = \bbeta C_e$ where $C_e \in \Fq^{r \times n}$ is a matrix over $\Fq$ which is the matrix of the coordinates of $\be$ with respect to the basis $\bbeta$.

Many decoding algorithms for rank metric codes consist of two major steps: the first step is to find the error support and recover one basis of the support, and the second is to use the basis to recover the matrix of the coordinates of the error.
Usually, once the error support is known, finding the matrix of the coordinates reduces to solving a linear system in $nr$ variables corresponding to the $nr$ entries of $C_e$.
In this paper, we will focus on the Rank Support Recovery problem for the LRPC codes.

\section{LRPC codes and their decoding} \label{Sec:LRPC}
LRPC codes were introduced in 2013 by Gaborit, Murat, Ruatta and Z\'{e}mor \cite{gaborit2013}.
\begin{definition}[LRPC code]\label{def:LRPC}
    Let $\calA \subseteq \Fqm$ be an $\Fq$-linear subspace of $\Fqm$ of dimension $d$ and $H \in \calA^{(n-k) \times n}$ a matrix of rank $n-k$.
    The code $\calC$ having $H$ as a parity check matrix is called an $[n,k]_{q^m}$ \textbf{LRPC code} of dual rank weight $d$.
\end{definition}
Due to their lack of a strong algebraic structure, these codes were proposed for several cryptographic applications \cite{RankSign2014, ROLLO, DuranteSiciliano}.
LRPC codes have a polynomial time decoding algorithm divided into two steps.
The first step aims to recover the error support; the second step uses the error support acquired in the first step to find the exact coordinates of the error.
Due to the page limit, in this paper we will focus only on the error support recovery part of the algorithm.

Consider $\calC \subseteq \Fqm^n$ be an LRPC code as in Definition \ref{def:LRPC}. 
Suppose we receive a message $\by = \bx + \be$, where $\calB = \Span{\Fq}{\be}$.
If $\dim(\calB) = r$ then $\calB$ will have a basis $\bbeta = (\Enum{\beta}{1}{r}) \in \Fqm^r$.
Consider now the syndrome $\bs = \be H^\intercal $, where $H \in \calA^{(n-k) \times n}$ for some space $\calA$ of dimension $d$.
Each component $s_i = \sum_{j = 1}^n e_j h_{i,j}$ of the syndrome is obtained as the sum of elements in the set $\calA \calE  = \{ e h \mid e \in \calE, h \in \calA \}.$
A key observation is that, if we denote with $\calA.\calE = \Span{\Fq}{\calA \calE}$ the smallest $\Fq$-linear vector subspace containing the set $\calA \calE$, we automatically have that $\bs \in (\calA.\calE)^{n-k}$.

Let $\balpha = (\Enum{\alpha}{1}{d})$ be a basis of $\calA$ and $\bb = (\Enum{\beta}{1}{r})$ be a basis of $\calB$.
Observe that for any $c = a e \in \calA.\calE$, we have 
$$
    c = \Big{(}\sum_{i=1}^da_i \alpha_i \Big{)}\Big{(}\sum_{j=1}^re_j \beta_j\Big{)} = \sum_{(i,j) \in [d] \times [r]} a_i e_j (\alpha_i \beta_i).
$$
As a consequence $\calA.\calE = \Span{\Fq}{\balpha \otimes \bbeta}$, where $\balpha \otimes \bbeta = (\alpha_1 \beta_1, \ldots, \alpha_d \beta_r)$ is a vector of length $rd$. 
This means that $\dim(\calA.\calE) \le \min(rd, m).$ In \cite{gaborit2013} the authors proved that the equality holds with a high probability.

After this observation, we are ready to present the error support recovery algorithm.
Since $\bs \in (\calA.\calE)^{n-k},$ we have that $\calS = \Span{\Fq}{\bs} \subseteq \calA.\calB.$
For $n-k \ge rd$, if we consider the $n-k$ components of $\bs$ as elements randomly extracted with a uniform distribution from $\calA.\calE$, with a good probability, we have $\calS = \calA.\calE$.
For the success of this algorithm it is crucial that $\calS = \calA.\calE$, therefore we need the condition $n-k \ge rd.$
The probability that $\calS = \calA.\calB$ is estimated to be at least $1 - q^{rd - (n-k)}$ \cite{ROLLO}.
This probability can be made arbitrarily small by choosing $n-k$ significantly larger than $rd$.
Notice that, if $\dim(\calA.\calB) = rd - t$, then we might require just $n-k \ge rd - t$. 
This means that in the case $\dim(\calA.\calB) < rd$, using the same number $n-k \ge rd$ of parity check equations, we will have an even better probability that $\calS = \calA.\calB.$
Therefore, considering $\dim(\calA.\calB) = rd$ can then be regarded as a worst-case scenario.

Suppose that $\calS = \calA.\calB$,
notice that, for any element $\alpha_i$ of the basis $\balpha,$ we have that $\calB \subset \calS \alpha^{-1}$.
If we intersect all these sets, we have that $\calB \subseteq \calS \alpha_{1}^{-1} \cap \ldots \cap \calS \alpha_{d}^{-1}$ where the equality holds with an estimated probability of at least $1-q^{-(d-1)(m-rd-r)}$ \cite[Prop. 2.4.2]{ROLLO}.
For large values of $m,$ this probability becomes quickly negligible.

Both of the failure probabilities reduce exponentially in $q$. 
The first probability of failure is harder to reduce than the second when $d>2$ since, increasing $n-k$ by one, improves the probability by a factor $q^{-1}$ while, increasing $m$ by one, improves
the second probability by a factor $q^{-(d -1)}.$

In \cite{LRPC-2019-TIT} Aragon, Gaborit, Hauteville, Ruatta, and Zémor gave two improved versions of the decoding algorithm that make use of two different expanding functions to be able to decode when $\calS \subsetneq \calA.\calE$.
In particular, they were able to decode when $\dim(S) = rd - c$, with $c < r$.
The main drawback of these two new algorithms is the need of a larger $m$ which has to be in the order of $3rd - 2$ in the first algorithm, and $2rd - r$ in the second.
Their first algorithm was able to decode LRPC codes with $n-k > (d-1)r$. The second improved considerably the success probability while keeping $n-k \ge rd$.

The algorithms we propose tackle the same problem. They offer  similar improvements while keeping $m$ in the order of $rd$.
We are able to decode when $c \leq d-2$ while asking $n-k \ge (r-1)d + 2$.

\section{Improved Error Support Recovery Algorithms}\label{Sec:decoder}
Consider an LRPC code defined over $\Fqm^n$ with parity check matrix $H \in \calA^{(n-k) \times n}$ where $\dim(\calA) = d$.
Suppose we receive $\by = \bx + \be$ where $\bx \in \calC$ and $\be \in \Fqm^n$ is an error of rank $r.$ 
Then the support $\calE = \Span{\Fq}{\be}$ has dimension $\dim(\calE) = r.$

We already showed $\bs = \be H^\intercal \in (\calA.\calE)^{n-k}$.
For the rest of this section let $\balpha = (\Enum{\alpha}{1}{d})$ be a basis of $\calA$ and $\bbeta = (\Enum{\beta}{1}{r})$ be a basis of $\calE$.
We will assume $\dim(\calA.\calE) = rd$ which is both the most common and the worst case, 

Suppose all of the $n-k$ elements of the syndrome are linearly independent and that $n-k = rd - c$.
The vector space $\calS = \Span{\Fq}{\bs}$ will be a proper subspace of $\calA.\calE$ of co-dimension $c$.
In \cite{LRPC-2019-TIT} it was shown how the code can still be decoded when $c < r$. 
The reason behind this condition was given by the following lemma \cite{LRPC-2019-TIT}.

\begin{lemma}\label{lm:r-c}
    Let $\calA,\calB$ be two subspaces of $\Fqm$ of dimension $d$ and $r$, respectively, and
     $\calS$ be a subspace of $\calA.\calB$ with $\dim(\calS)=\dim(\calA.\calB)-c$.
    For any nonzero elements $a\in \calA$, $b\in\calB$, we have
    $
    \dim(\calS a^{-1} \cap \calB) \geq r - c$ and $\dim(\calS b^{-1} \cap \calA) \geq d-c.
    $
    
\end{lemma}
\begin{proof}
    Note that $\calS a^{-1} +  \calB = (\calS +  \calB a) a^{-1} \subseteq (\calA.\calB) a^{-1}$. Thus $\dim(\calS a^{-1} +  \calB) \le \dim(\calA.\calB).$
    From this fact, we have the following inequality
    \begin{equation*}
    \begin{aligned}
        \dim(\calS a^{-1} \cap \calB) & =  \dim(\calS a^{-1}) +  \dim(\calB) - \dim(\calS a^{-1} + \calB) \\
                                 & \ge   \dim(\calA.\calB) - c + r - \dim(\calA.\calB) 
        = r-c.
    \end{aligned}
\end{equation*}
 Due to the symmetrical role of $\calA$ and $\calB$, the second statement follows similarly.
\end{proof}
Lemma \ref{lm:r-c} shows that, when $c < r$, many elements of $\calB$ are contained in $\calS a^{-1}$ for every nonzero element $a$ in $\calA$.

The original algorithm considers only the $d$ sets of the form $\calS \alpha_i^{-1}$ where $\balpha = (\Enum{\alpha}{1}{d})$ is a basis of $\calA$.
In our case, since different subsets of $\calE$ might be contained in different sets, we want to use as many different sets as we can.
To better understand why considering the elements of a basis might not be enough, consider the following example.
Let $s_1 = \alpha_1 \beta_1 + \alpha_2 \beta_2$ and $s_2 = \alpha_1 \beta_2 + \alpha_2 \beta_1$.
If we consider the set $\calS = \Span{\Fq}{s_1, s_2}$ we have that $s_1 + s_2 = (\alpha_1 + \alpha_2)(\beta_1 + \beta_2) \in \calS$ then $(\beta_1 + \beta_2) \in \calS (\alpha_1 + \alpha_2)^{-1} \cap \calE$ while $\calS \alpha_1^{-1} \cap \calE$ and $\calS \alpha_2^{-1} \cap \calE$ are both empty.

Intuitively, if we consider all the sets of the form $S a^{-1}$ for some $a \in \calA^*,$ the elements of $\calE \cap S a^{-1}$ will appear in many of the other sets with the same form, while the elements of $\calS a^{-1}\setminus \calE,$ will occur with significantly less frequency in the other sets.

The following theorem gives a simple way to count how many of the sets in the form $\calS a^{-1}$ for $a \in \calA^*$
contain an element $x \in \Fqm$.
\begin{theorem}\label{th:SxCapA} 
    Let $\calS$ be a subspace of $\calA.\calE$, where  
    $\calA,\calE \subseteq \Fqm$ are two subspaces of $\Fqm$ of dimension $d$ and $r$.
    Define a multi-set $Z$ as a union of $\calS a^{-1}$ for all $a\in \calA^*$, i.e., 
    \begin{equation}
        Z = \biguplus_{a \in A^*} \calS a^{-1}.
    \end{equation}
    The multiplicity of $x \in \Fqm^*$ in $Z$ is given by 
    \begin{equation}
        \Mul(x,Z) = |\calS x^{-1} \cap \calA^*| = q^{\dim(\calS x^{-1} \cap \calA)} - 1
    \end{equation}
    and $\Mul(x=0,Z) = q^d - 1$.
\end{theorem}
\begin{proof}
    Consider the set 
    $$\alpha(x) = \{ a \in \calA^* \mid \exists s \in \calS : x = s a^{-1} \}.$$
    By definition of $Z$, we have that $|\alpha(x)| = \Mul(x, Z)$.
    For $x = 0$, since $0 \in \calS,$ we have that $0 = 0 a^{-1}$ for all possible $a \in \calA^*$ therefore $\alpha(0) = \calA^*$ and $\Mul(0,Z) = |\calA^*| = q^{d} - 1$.
    To complete the proof of the theorem it is sufficient to show that $\alpha(x) = \calS x^{-1} \cap \calA^*$.
    Notice that, for $a \in \calA^*, 0 \neq x \in \calS$ we have
    \begin{equation*}
       x = s a^{-1} \iff a = s x^{-1}.
    \end{equation*}
    By the definition of $\alpha(x)$ and $\calS x^{-1}$,
    the first equation is equivalent to $a \in \alpha(x)$ while the second equation is equivalent to $a \in \calS x^{-1} \cap \calA^*.$ 
    This shows that $\alpha(x) = \calS x^{-1} \cap \calA^*$.
    We know that $|\calS x^{-1} \cap \calA| = q^{\dim(\calS x^{-1} \cap \calA)}$. The desired conclusion thus follows.
\end{proof}
The following corollaries characterize the multiplicities of different elements in $Z$.

\begin{corollary}
    Let $\calA,\calB,\calS$ and $Z$ be as in Theorem \ref{th:SxCapA}, where $\calS$ has dimension $rd - c$. Let 
    $w=\max\{1, rd-c+d-m\}$. Then
    for any $x \in Z \setminus \calB$, its multiplicity  $\Mul(x,Z) \ge q^w - 1.$
\end{corollary}\label{co:SxCapA}
\begin{proof}
    Since $x \in Z$, there exists $a \in \calA$ such that $x = s a^{-1}$ for some $s \in \calS$, implying $a = s x^{-1} \in \calS x^{-1}$. Thus  $\dim(\calS x^{-1} \cap \calA) \ge 1$.
    In addition, we know that 
    \begin{align*}
        \dim(\calS x^{-1} \cap \calA) 
        & = \dim(\calS x^{-1}) + \dim(\calA) - \dim(\calS x^{-1} + \calA)
        \\& = \dim(\calS) +\dim(\calA)  - \dim(\calS x^{-1} + \calA)
        \\&\geq rd-c+d -m
    \end{align*} 
    since $\calS x^{-1} + \calA$ has dimension at most $m$.
\end{proof}

\begin{corollary}\label{co:SbCapA}
Let $\calA,\calB,\calS$ and $Z$ be as in Theorem \ref{th:SxCapA}, where $\calS$ has dimension $rd - c$. 
    Then $\Mul(b,Z) \ge q^{d-c} - 1$, $\forall b \in \calB$.
\end{corollary}
\begin{proof}
    It is clear that $Mul(0,Z) = q^d - 1 \ge q^{d-c} - 1$.
    For $b \in \calB^*$, it follows from Theorem \ref{th:SxCapA} that
    \begin{equation}\label{eq:Mul(b,Z)}
        \Mul(b,Z) = |\calS b^{-1} \cap \calA^*| = |\calS b^{-1} \cap \calA|
        - 1.
    \end{equation}
        By
    Lemma \ref{lm:r-c} we know that $|\calS b^{-1} \cap \calA| \ge q^{d-c}$.
    This together with the above equality leads to 
    the desired statement $\Mul(b,Z) \ge q^{d-c} - 1.$
\end{proof}

Our main goal is to recover the error support $\calB$ when the support $\calS$ is a proper subspace of $\calA.\calB$. 
From the above analysis, we can create the multi-set $Z$ and focus only on the elements with multiplicity greater than or equal to $ q^{d - c} - 1.$
Consider the set 
\begin{equation}
    \Tilde{\calB} = \{x \in Z \mid \Mul(x,Z) \ge q^{d-c} - 1 \}.
\end{equation} Thanks to Corollary \ref{co:SbCapA}, we know that $\calB \subseteq \Tilde{\calB}.$ 
From Corollary \ref{co:SxCapA}, it is better to choose $m \ge rd - c + d-1$ such that the generic element of $Z$ can have multiplicity as low as $q-1$ while the elements of $\calB$ will always have multiplicity at least $q^{d-c} - 1$.
It is clear that $d-c \geq 2$ is a minimal condition to distinguish a generic element in $Z$ from the elements of $\calE$ when $q=2$.

It is possible that some elements of $Z$ have large multiplicity even if they are not elements of $\calE$, when that happens we have $|\Tilde{\calE}| > |\calE| = q^r$. Assume that $\Tilde{\calE} = \calE \cup X \subseteq \Fqm$ where $X$ is a set of small cardinality $|X| < q^r$. Note that for any $x\in \calE$, $\calE \subset (x + \Tilde{\calE}) \cap\Tilde{\calE}$. With this fact, we can quickly obtain $\calE$ from $\Tilde{\calE}$ 
in the following way: take a random $x \in \Tilde{\calE}$, if $|\Tilde{\calE} \cap (\Tilde{\calE} + x)| > q^r$, then take $\Tilde{\calE} = \Tilde{\calE} \cap (\Tilde{\calE} + x)$, and continue this process until $|\Tilde{\calE}| = |{\calE}|= q^r.$ Such a process of filtering works well when the number of outliers is small compared to the size of $\calE$. 

The above process of selecting elements in $Z$ with multiplicites at least $q^{d-c}-1$ 
is summarized in \ref{Algorithm 1}. 
Empirically, setting a large enough $m$, we will directly obtain $\Tilde{\calE }= \calE$.
Diminishing the value of $m$, the set $\Tilde{\calE }\setminus \calE$ will progressively grow until it will not be possible anymore to retrieve the correct $\calE$.



\begin{algorithm}[http]
		\small
		\SetAlgoLined
		\KwIn{A parity check matrix $H \in \calA^{(n-k)\times n}$ where $\dim(\calA) = d$ and $\calS = \Span{\Fq}{\by H^\intercal}$ of dimension $rd-c$ where $\by = \bx + \be \in \Fqm^n$, $\bx \in \calC$ and $\be$ is an error of rank $r$.}
        \KwOut{The support $\calE = \Span{\Fq}{\be}$ of dimension $r$.}
        \tcp{Assumption: $\dim(\calA.\calE) = rd$}
        \uIf{$rd - \dim(\calS) < d - 2$}
        {   $c= rd - \dim(\calS)$\;
            $Z = \{* \, \,*\}$   \tcp*{Create an empty multi-set}
            \For{$a \in \calA^*$}
            {
                \For{$s \in \calS$}
                {
                    
                    $Z = Z \uplus \{s a^{-1}\}$ \tcp*{Compute $Z$} 
                }
            }
            $\calE = \{ \}$ \;
            \For{$z \in Z$}
            {
                \If{$\Mul(z, Z) \ge q^{d-c} -1$}
                {
                    $\calE = \calE \cup \{z \}$ \;
                }
            }
            Return $\calE$\;
        }
        \tcp{The dimension of $\calS$ is too low}
        \Else 
        { 
            Error Support Recovery Failure\;
        }
		\caption{Error support recovery of LRPC codes (Theoretical Version)}\label{Algorithm 1}
		\normalsize
	\end{algorithm}

In \ref{Algorithm 1}, the generation of the multi-set $Z = \biguplus_{a \in \calA^*} \calS a^{-1}$ has a high time and space complexity.
In order to address this drawback, we propose a more practical algorithm, which select elements from the intersection of $t>2$ subspaces $\calS a^{-1}$ for some rounds and generates the error support generated by those elements.
We first give some theoretical analysis before presenting the second algorithm. 
\begin{proposition}\label{Prop}
Let $\calS\subset \calA.\calB$ be given as in Theorem \ref{th:SxCapA} with dimension $rd - c$. Take nonzero elements $a_1, \dots, a_t$ from $\calA$.
Then 
$
\dim(\calS a_1^{-1} \cap \cdots \cap \calS a_t^{-1} \cap \calB) \geq r-tc.
$
\end{proposition}
\begin{proof}
Take $\calT_i=\calS a_1^{-1} \cap E$ of dimension $\ge r-c$. Then
\begin{align*}
    &\dim\big(\calT_1\cap \cdots \cap \calT_t \big) 
     \\  = &\dim(\calT_2 \cap \cdots \cap \calT_t \big) + \dim\big(\calT_1)- \dim(\calT_1 + (\calT_2 \cap \dots \cap \calT_t))
    \\\geq & \dim(\calT_2 \cap \cdots \cap \calT_t \big) + (r-c) - r = \dim(\calT_2 \cap \cdots \cap \calT_t \big)-c.
\end{align*}
Iterating this process on $t$ leads to the desired inequality.
\end{proof}
By Proposition \ref{Prop}, the intersection of $t$ subspace $\calS a_i^{-1}$ may contribute to $r-tc$ independent elements in $\calE$. This implies that we can recover the error support $\calE$ by accumulating elements from such intersections.  
The following result is obtained by applying 
\cite[Prop. 2.4.2]{ROLLO}.
\begin{proposition} \label{prop3}
    Consider $\calE$ a subspace of dimension $r$. Let $\Omega$ be the set of all subspaces of dimension $rd-c$ having intersection of dimension at least $r-c$ with $\calE$.
    The probability that the intersection of $t$ subspaces $\Enum{\calS}{1}{t}$ independently chosen uniformly at random in $\Omega$ contains some elements not in $\calE$ is approximated by
    $$
    \text{Prob}((\calS_1\cap \cdots \cap \calS_t) \subset \calE) )  \approx 1 - 1/q^{(t-1)m+r - trd}.
    $$
\end{proposition}
\begin{proof}
    Following the proof of \cite[Prop. 2.4.2]{ROLLO}.
    The subspaces $\calS_1, \ldots, \calS_t$ are not independent since each contains part of $\calE$.
    If consider the quotient space $V = \Fqm / \calE$ of dimension $m-r.$ 
    The sets $\calR_i = \calS_i / \calE$ are now independent and will have dimension $\dim(\calR_i) = \dim(\calS_i) - \dim(\calS_i \cap \calE) \le rd - c - r + c = rd - r.$
       Fix $0 \neq y \in \calR_1$, if we consider $\calR_2$ independent from $\calR_1$ the probability that $y \in \calR_2$ is $(|\calR_2| - 1)/(|\Fqm / \calE| - 1) \leq (q^{rd-r} - 1)/(q^{m - r} - 1).$ 
    The same will be true for the other $\calR_i$ which gives a probability of $((q^{rd-r} - 1)/(q^{m-r} - 1))^{t-1}.$
    If we consider each $y$ as independent we have to multiply this probability by the number of $y$ which is $|\calR_1| - 1 = q^{rd- r} - 1$
    This gives us
    \begin{align*}
         \text{Prob}(\dim(\cap_{i \in [t]}\calR_i) > 0) &\le (q^{rd - r} - 1)\left( \frac{q^{rd-r}-1}{q^{m-r} -1} \right)^{(t-1)} \\
        & \approx q^{-t(m -rd) + m - r}
    \end{align*}
    Notice that $\text{Prob}(\dim(\calR_1 \cap \ldots \cap \calR_t) > 0)$ is equal to $\text{Prob}((\calS_1 \cap \ldots \cap \calS_t) \setminus \calE \neq \{0\})$. This ends the proof.
\end{proof}
With the statements in Propositions \ref{Prop} and \ref{prop3}, we propose a more practical decoding of LRPC codes in \ref{Algorithm 2}.
Even though the subspaces $\calS a^{-1}$ in \ref{Algorithm 2} cannot be considered as uniformly independently chosen, heuristically they seem to follow Proposition \ref{prop3}.

\begin{algorithm}[http]
		\small
		\SetAlgoLined
		\KwIn{A parity check matrix $H \in \calA^{(n-k)\times n}$ where $\dim(\calA) = d$ and $\calS = \Span{\Fq}{\by H^\intercal}$ of dimension $rd-c$ where $\by = \bx + \be \in \Fqm^n$, $\bx \in \calC$ and $\be$ is an error of rank $r$.}
        \KwOut{The support $\calE = \Span{\Fq}{\be}$ of dimension $r$.}
        \tcp{Assumption: $\dim(\calA.\calB) = rd$}
        \uIf{$rd - \dim(\calS) < d$}
        {
           $S = \{s_1, \dots, s_u\}= \text{Basis}(\calS)$\;
           $t = q^{\lceil\log_q(r/c)\rceil}$\;
           $\calB= \{ \,\}$\;
           \While{$\dim(\Span{\Fq}{\calB}) < r$}
           {
             \tcp{Generate $t$ random elements from $\calA^*$} 
                $Y =\{a_1, \dots, a_t\}=$Random$(\calA^*, t)$\;
                \For{$a \in Y$}
                {
                    Generate $\calS a^{-1}=\Span{\F_q}{\{a^{-1}s_1, \dots, a^{-1}s_u\}}$\; 
                }
                $\calB = \calB + \bigcap_{1\leq i\leq t} \calS a_i^{-1}$\;
           }
            Return $\calB$ \;
        } 
        \Else
        {
            Return "Support recovery failure" \;
        }
       		\caption{Error support recovery of LRPC codes}\label{Algorithm 2}
		\normalsize
\end{algorithm}

\section{Discussion} \label{Sec:dis}

\begin{table}[t!]
\renewcommand{\arraystretch}{1.2}
\begin{center}
\begin{scriptsize}
\begin{tabular}{|c|c|c|c|l||c|c|c|l|}
\hline
$c$ &$r, d$ & $t$ & $m$ & Success & $r, d$ & $t$ & $m$ & Success\\
\hline 
\multirow{2}{*}{1} & $5, 5$ & $4$ & $40$ & $99.9\%$ & $5, 6$ & $4$ & $46$ & $99.4\%$ \\ \cline{2-9}
 &  $5, 5$ & $4$ & $41$ & $100\%$ & $5, 6$ & $4$ & $47$ & $100\%$  \\ \hline
 \multirow{2}{*}{2} & $5, 5$ & $4$ & $42$ & $99.9\%$ & $5, 6$ & $4$ & $48$ & $99.9\%$ \\ \cline{2-9}
 &  $5, 5$ & $4$ & $43$ & $100\%$ & $5, 6$ & $4$ & $49$ & $100\%$  \\ \hline
\end{tabular}
\end{scriptsize}
\end{center}
\caption{Success rate of \ref{Algorithm 2} with $n-k=rd-c$, $c=1,2$}\label{tb:error1}
\renewcommand{\arraystretch}{1}
\end{table}

For the two improved decoding algorithms in this paper, \ref{Algorithm 1} utilizes the observation that all nonzero elements in the parity-check support $\calA$ can contribute to the process of error support recovery, yet it can only work well for small parameters due to its high complexity in memory. \ref{Algorithm 2} derives elements in the error support by repeatedly intersecting $t>2$ subspaces $\calS a_i^{-1}$. In this way, the algorithm is pretty efficient for different choices of parameters $m, r, d, c$, and $t$. 

Our algorithms and the algorithms in \cite{LRPC-2019-TIT} both improved the original decoding of LRPC codes for the cases where the syndrome support has dimension $rd-c$ for $c>0$. 
While the algorithms in \cite{LRPC-2019-TIT} require $m\geq 3rd-2$ and $m\geq 2rd-r$, respectively, our new algorithms have significantly loosened the requirements on $m$, namely, 
\ref{Algorithm 1} requires $m \ge rd - 2(d-c)$ and \ref{Algorithm 2} requires $m \ge \frac{t}{t-1} rd$. Increasing the value of $m$ by $1$ will improve the failure rate by $q^{-(t-1)}$. 

In Table \ref{tb:error1} we provide some experimental results for \ref{Algorithm 2}, where
we run a series of 1000 experiments for parameters $n-k=rd - c$ for $c=1, 2$, $(r,d)=(5,5), (5,6)$, $t=4$
for different $m$, and report the lower bound of $m$  that gives nearly 100\% success rates of error support recovery. 
The parameters that rule the probability of success are $c $ and $m$. In our experiments we used $k = 1$ and $n = rd - c + k$.
We observe that repeating the same experiments for different values of $k$ did not affect the results reported in Table \ref{tb:error1}.

\newpage

\bibliographystyle{ieeetr}
\bibliography{RankMetricCodes.bib}


\end{document}